\documentclass{llncs}
\usepackage{amsfonts}
\usepackage{epsfig}
\usepackage[usenames]{color}
\usepackage[vlined]{algorithm2e}
\usepackage{algorithmic}

\newtheorem{cl}{Claim}

\title{A Distributed Enumeration Algorithm and Applications to All Pairs 
Shortest Paths, Diameter...}
\author{Y. M\'etivier, J.M. Robson, and A. Zemmari}
\institute{Universit\'e de Bordeaux, LaBRI, UMR CNRS 5800\\ 351 cours de la
  Lib\'eration, 33405 Talence, France\\ 
\{metivier, robson, zemmari\}@labri.fr } 
\date{ } 
\begin{document}
\maketitle
\begin{abstract}
We consider the standard message passing model; we assume the system is fully
synchronous: all processes start at the same time
 and time proceeds in synchronised rounds. In each round each vertex
can transmit a different message of size $O(1)$ to each of its neighbours.
This paper proposes and analyses a distributed enumeration algorithm 
of vertices of a graph  having a distinguished vertex
which satisfies that two vertices with  consecutive numbers are at
distance at most $3$. 
We prove that its time complexity is $O(n)$ 
where $n$ is the number of vertices of the graph.
Furthermore, the size of each message is $O(1)$ thus its bit complexity is
also $O(n).$ 
We provide some links between this enumeration and Hamiltonian graphs
from which we deduce
 that this enumeration is optimal in the sense that there
does not exist an enumeration 
which satisfies that two vertices with  consecutive numbers are at
distance at most $2$.

We deduce   from this enumeration
algorithms which compute  all pairs shortest paths
and the diameter with a time complexity and a bit complexity equal to $O(n)$.
This improves the  best known distributed 
algorithms (under the same hypotheses) for computing all pairs shortest paths
or the diameter
presented  in \cite{PRT12,HW12}   having a time 
complexity equal to  $O(n)$ and which  use messages of size $O(\log n)$ bits.

\end{abstract}

\section{Introduction}
\subsection{The problem}
In this paper we consider the all pairs shortest paths problem 
in a distributed network. We assume that there exists 
a distinguished vertex (called $Leader$) so that
 there exist
distributed algorithms for solving it. 
We are interested in optimal solutions in time and in number of bits for this
problem.

Distributed algorithms for solving 
the all pairs shortest paths problem find extensive use in  communication
networks and thus in distributed computing.
For example,  the shortest path between a source and a destination
is considered as the most economic.
Finally, many routing schemes use shortest paths and designing
such a scheme consists of computing the shortest routes and storing
information on vertices for routing messages (\cite{Peleg}, p. 105).

The solution presented in this paper is based on a non trivial
distributed  enumeration algorithm  which satisfies that two vertices
with consecutive numbers are at distance at most $3$.

\subsection{The {M}odel}

\subsubsection{The Network.}
We consider the standard message passing model for distributed
computing.  The communication model consists of a point-to-point
communication network described by a   connected graph 
$G=(V(G),E(G))$ $(=(V,E)$ for short)
where the vertices $V$ represent network processes and the edges $E$
represent bidirectional communication channels. Processes communicate
by message passing: a process sends a message to another by depositing
the message in the corresponding channel. In the sequel, 
we consider only connected graphs.
We assume the system is
fully synchronous, namely, all processes start at the same time
and time proceeds in synchronised rounds.

\subsubsection{Time Complexity.}
A round (cycle) of each process is composed of the following three
steps: 1. Send messages to (some of) the neighbours, 2. Receive
messages from (some of) the neighbours, 3. Perform some local
computation.  As usual the time
complexity is the number of rounds needed until every vertex has
completed its computation.

\subsubsection{Bit Complexity.}
We follow the definition given in \cite{KOSS}.
By definition, in a bit round each vertex can send/receive at most $1$ bit 
from each of its neighbours.
The bit complexity of  algorithm $\cal A$
is the number of bit rounds to complete algorithm $\cal A.$

\begin{remark}
A round of an algorithm consists of $1$  or more bit rounds. 
 The bit complexity
of a distributed  algorithm is an upper bound on the total number 
of bits exchanged per channel during its execution. It is also an upper bound
on its time complexity.

If we consider a distributed algorithm having messages of size $O(1)$ 
(and this is the case in this paper)
then the time complexity and the bit complexity
are equal modulo a multiplicative constant.
\end{remark}
The bit complexity is considered as a finer measure of
communication complexity and it has been studied for breaking symmetry
or for colouring in \cite{BNNN90,BMW} or in \cite{KOSS,DMR}.  Dinitz
et al.  explain in \cite{DMR} that it may be viewed as a natural
extension of communication complexity (introduced by Yao \cite{Yao})
to the analysis of tasks in a distributed setting. An introduction to
this area can be found in Kushilevitz and Nisan \cite{KN}.

%\subsubsection{Single Bit Messages.}
%Classically, there are two models for the size of the messages: the
%{\cal {LOCAL}} model and the {\cal CONGEST} model (see \cite{Peleg}
%p. 27). The first one allows message of unlimited size while the second
%allows messages with a size bounded by $O(\log n)$ ($n$ is the size of
%the network).  In both models vertices have unique identifiers.
%
%The model of this paper is anonymous (no uniqueness of identifiers for
%the nodes) and nodes have no global knowledge on the network such as
%its size.  Thus, in this context, when a processor builds a message
%its size cannot depend on the size of the network and it is natural to
%consider (and we consider)
%single bit messages or more generally messages with bounded
%sizes.

\subsubsection{Network and Processes Knowledge.}

The network $G=(V,E)$ is anonymous: unique identities are not available to
distinguish the processes.  We only assume that there is an 
elected (a distinguished) vertex denoted $Leader$.
We do not assume any global knowledge of
the network, not even its size or an upper bound on its size. The
processes do not require any position or distance information. 
Each process knows from which channel it receives or to which channel
it sends a message,
thus one supposes that the network is represented by a
connected  graph with a port numbering function 
defined as follows (where $I_G(u)$ denotes the set of edges of $G$
incident to $u$):
\begin{definition}
Given a 
graph $G=(V,E)$, a \emph{port numbering} function $\delta$ is
a set of local functions $\{\delta_u \mid u \in V\}$ such that for
each vertex $u \in V$, $\delta_u$ is a bijection between $I_G(u)$ and
$[1,\deg_G(u)]$. 
\end{definition}
\subsubsection{All Pairs Shortest Paths, Diameter,  Girth, Cut-Edge and
Cut-Vertex.}
We follow definitions given in \cite{Rosen}.
A walk in a graph $G=(V,E)$ is a finite alternating sequence of
vertices and edges, beginning and ending with a vertex and where each edge is 
incident with the vertices immediately preceding and following it.
A trail is a walk in which no edge occurs more that once. A path
is a trail in which all of its vertices are different, except that the initial
and final vertices may be the same.
A walk with at least $3$ vertices in which the first and last vertices are 
the same but all other vertices are distinct is called a  cycle.

Let $G=(V,E)$ be a connected graph, let $u, v \in V$. The 
distance between $u$ and $v$  in $G$, denoted $dist_G(u,v)$, 
is the length of a shortest path between $u$ and $v$ in $G$.

Given a vertex $v$ of a connected graph, the eccentricity of
$v$ is the greatest distance from $v$ to another vertex.

The all pairs shortest paths (APSP for short) problem in $G$ is to
compute  the length of shortest paths between any pair of vertices in $G.$

The diameter of $G$, denoted $D(G)$, 
is the maximum distance between any two vertices of $G,$ 
i.e., $D(G) = \max \{dist_G(u,v) \mid u, v \in V\}$.

The girth of a graph $G$ is the length of a shortest cycle of $G.$

A cut-vertex is a vertex whose removal increases the number of 
connected components.

A cut-edge is an edge whose removal increases the number of connected
components.

We  use trees and we
follow the presentation given in \cite{CLR}. 
A tree is a connected acyclic graph. A rooted tree is a tree in which one
of the vertices is distinguished from the others (called $Leader$ in this
work).
A spanning-tree of a connected graph $G=(V,E)$ is a tree $T=(V,E')$ such that
$E'\subseteq E$.
\subsection{Our Contribution}

We present a distributed enumeration algorithm, denoted $DEA$,
which assigns to each vertex of a graph $G$ of size $n$
having a distinguished vertex, denoted $Leader$,
a unique integer of  $\{1, 2, \dots, n\}$
such that the distance between any 
two vertices having
two consecutive numbers is at most $3$. 
 This algorithm uses messages of
size $O(1)$ and has a time complexity equal to $O(n)$.

The steps of Algorithm $DEA$  are:
\begin{enumerate}
\item computation of a Breadth-First-Search (BFS)
 spanning-tree of 
$G$ whose  root is $Leader$;
\item enumeration of the vertices with respect to a special traversal
      of the BFS spanning-tree.
\end{enumerate}

This enumeration enables the initialisation
of  anonymous 
waves, with respect to the enumeration order, i.e., the first wave
is initialised by the vertex numbered $1$, the second wave by the vertex
numbered $2$, etc.
Anonymous waves reach vertices with respect to the enumeration order 
(i.e., the wave initialised by the vertex numbered $i$ reaches any vertex
after the wave initialised by the vertex numbered $i-1$ and before the wave
initialised by the vertex numbered $i+1$)
and thus
 are implicitly identified. This fact
allows each vertex to 
 compute  its distance to any vertex without the computation and the
use of the distance itself 
but by inference from the time; in this way  all pairs shortest paths 
are obtained in time $O(n)$, each message having
a constant size so that the bit complexity 
 is $O(n)$.

We deduce also a
 distributed algorithm for graph diameter with a bit complexity and a time
complexity equal to $O(n)$.

Frischknecht et al. proved (\cite{FHW12}, Theorem 5.1)
 that: ``For any $n\geq 10$ 
and $B\geq 1$ and sufficiently small $\epsilon$ any distributed
randomized $\epsilon$-error algorithm A that computes the exact diameter
of a graph
 requires at least $\Omega(n/B)$ time for
some $n$-node graph even when the diameter is at most $5$,''
where $B$ is the size of messages.

From this result we deduce that the bit complexity of our algorithm is optimal
and the time complexity is also optimal for messages of size $O(1)$.

In the remainder of this work, we explain 
how the enumeration algorithm can be applied for computing the girth,
cut-edges, cut-vertices or for recognising biconnected graphs.

\begin{remark}\label{optimal}
We may wonder whether it is possible to obtain an enumeration of vertices
such that the distance between any two vertices having two consecutive numbers 
is at most $2$. We explain in the next section why the answer
is negative. It indicates that in some certain sense our enumeration 
is optimal.
\end{remark}
\subsection{Related Work: Comparisons and Comments}

\subsubsection{Enumeration Algorithm.}
The enumeration of vertices of a connected graph such that two  consecutive
vertices of the enumeration are at distance at most $3$ is also
presented in \cite{Sekanina}; this paper
presents  a sequential algorithm for computing  such an enumeration.

Let $G$ be a graph.
A Hamiltonian path in $G$ is a path that includes all the vertices of $G$.

We recall that the cube of a graph $G$, denoted $G^3$, is 
 the graph with the set of vertices of $G$ in which there is an edge 
between two vertices $u$ and $v$  if the distance between $u$ and $v$ in
$G$ is at most
$3$.
It was noticed by C. Gavoille \cite{G14} that the existence
of such an  enumeration is equivalent to the fact that the cube of a connected
graph $G$  contains a Hamiltonian path

A 
 cycle
containing all vertices of $G$ is called a Hamiltonian cycle of $G$,
and $G$ is called a Hamiltonian graph.
From our enumeration result we deduce a well known result
\cite{CK}:
\begin{theorem}
If $G$ is a connected graph then $G^3$ is a Hamiltonian graph.
\end{theorem}

As for the cube of a graph, 
the square  of a graph $G$, denoted $G^2$, is 
 the graph with the set of vertices of $G$ in which there is an edge 
between two vertices $u$ and $v$  if the distance between $u$ and $v$ in
$G$ is at most
$2$.

The previous theorem is no longer true for the square of a tree
as indicated by the next theorem.
Let $K_{1,3}$ be the tree with one internal vertex and three leaves.
Let $S(K_{1,3})$ be the subdivision of $K_{1,3}$ formed by inserting
a vertex of degree two on each edge of $K_{1,3}$.
A graph $G'=(V',E')$ is called a subgraph of a graph $G=(V,E)$
if $V'\subseteq V$, $E'\subseteq E$ and $V'$ contains all the 
endpoints of the edges in $E'$.
Regarding the characterisation of trees with Hamiltonian square, Harary 
and Schwenk \cite{Harary} proved that:
\begin{theorem}
Let $T$ be a tree with at least $3$ vertices. $T^2$ is a Hamiltonian
graph if and only if $T$ does not contain $S(K_{1,3})$ as a subgraph.
\end{theorem}

In fact for our work, a priori, we only need  a Hamiltonian path.
Remark \ref{optimal} is a direct consequence of the 
following result.
In \cite{RR11}, it is proved that the square of a tree $T$
has a Hamiltonian path if and only if $T$ is a horsetail.
The definition
of a horsetail is rather technical thus we do not provide it; in our context
the important fact is that the family of trees which are not
horsetails is infinite.

\subsubsection{All Shortest paths, Diameter, Girth...}
The computation of all pairs shortest paths, of the diameter or of the girth
is the subject of many studies. Very complete recent surveys on these
 questions and on associated results can be found in 
\cite{FHW12,HW12,PRT12,LP13,N14}.

Known results depend on the size of the messages, denoted $B$ in the sequel,
 that a vertex can transmit
to its neighbours. Furthermore it depends also on whether algorithms compute
exact values or approximations (approximations enable in some cases an
improvement in
the running time).
Among the most recent results one can cite:
\cite{WW10,AB11,RT11,RW11,RW12,FHW12} \cite{HW12,PRT12},
\cite{LP13,N14}.

Frischknecht et al. 
\cite{FHW12} established an $\Omega(n/B)$ 
 lower bound for the number of communication rounds needed for computing
the diameter of a graph (they use a non-trivial technique of transferring
lower bounds from communication complexity and graph-constructions).
Thus the challenge
for the computation of all pairs shortest paths or of the diameter
 in linear time (in this context) concerns the size of messages.
Almeida et al. present in \cite{AB11} an algorithm with a time
complexity $O(D)$ (where $D$ is the diameter of the graph)
with large messages: $B=O(n\log n)$.
The best known distributed algorithms for computing the diameter with 
$B=O(\log n)$
are presented  in \cite{PRT12,HW12}.
Both assume that the size $n$
of the graph is known and
each vertex has
a unique identifier from $\{1,...,n\}$.
In both cases, algorithms compute  BFS spanning-trees rooted at each vertex and compute
distances between any two vertices. The time complexity of both
algorithms is $O(n)$.
The key point is that there is no collision between messages
of different BFS spanning-trees construction processes: 
at any time a vertex is active for 
the construction of at most one BFS spanning-tree. 
Messages enable the computation of 
distances between vertices so that the size of messages 
is $O(\log n)$ and the bit complexity
of both algorithms is $O(n \log n).$ 

\begin{remark}
Our initial knowledge and hypotheses 
on graphs are equivalent to the initial knowledge and initial hypotheses
on  graphs in \cite{PRT12,HW12} in the sense that one can be 
obtained from the other in a linear time with a linear bit complexity.
\end{remark}

\par
\noindent
\vspace{1.5cm}
\begin{minipage}{\linewidth}
\centering
\begin{tabular}{|p{5.cm}|p{0.9cm}|p{4.8cm}|p{2.4cm}|}
%\begin{tabular}{|p{2.3cm}|p{2.9cm}|p{3.6cm}|p{2.5cm}|p{3.6cm}|}
%\begin{tabular}{|l|l|l|l|l|}
\hline

& Time & Message size (number of bits) & bit complexity \\

\hline

Almeida et al. \cite{AB11}&  $O(D)$ &  $O(n\log n)$ & $O(Dn\log n)$ \\
\hline
Holzer and Wattenhofer \cite{HW12} & $O(n)$   & $O(\log n)$ &$O(n\log n)$ \\
\hline
Peleg et al. \cite{PRT12} &  $O(n)$ & $O(\log n)$& $O(n\log n)$\\
\hline
This paper & $O(n)$ & $O(1)$ & $O(n)$\\

\hline
\end{tabular}\par
\vspace{0.5cm}
This table  summarises the comparison between the complexities
of various diameter\\ 
algorithms and the complexities of the 
diameter algorithm presented in this paper.

\bigskip

\end{minipage}
\vspace{0.5cm}

General considerations and
results concerning  cut-edges, cut-vertices and biconnectivity
are presented in
\cite{PritchardT11,Chaudhuri98,Hohberg90}.
Thurimella \cite{Thurimella97} proved that the diameter is a more
precise parameter for the time complexity of finding
cut-edges, cut-vertices or deciding the biconnectivity, more precisely,
Thurimella obtained time complexity  $O(D + \sqrt n \log^* n)$
for these problems  on a graph $G$ where $D$ is the diameter of $G.$

This paper is organised as follows. Section 2 presents a distributed
enumeration algorithm,
denoted $DEA$,
 and proves that its time complexity and its bit
complexity are linear. Section 3 applies Algorithm DEA
to all pairs shortest paths. Section 4 gives an immediate application
for computing the diameter.
Section 5
explains how anonymous waves
enable the computation of the girth, cut-edges and cut-vertices
with a linear time complexity and a linear  bit complexity.
\section{A Distributed Enumeration Algorithm} 
This section describes the steps of Algorithm $DEA$ which 
enumerates vertices of a given graph $G$ having a distinguished
vertex, denoted $Leader$,:
\begin{enumerate}
\item computation of a BFS spanning-tree of $G$ whose  root is $Leader$;
\item enumeration of the vertices with respect to a special traversal
      of the BFS spanning-tree.
\end{enumerate}

Consider a vertex $v$ in a rooted tree with root $Leader$.
The length of the unique path from $Leader$ to  $v$ is
the level of $v$.
 Any node
$w$ on the unique path from $Leader$ to $v$ is called an ancestor of $v$.
If the last edge on the unique path from $Leader$ to a vertex $v$
is $\{w,v\}$ then $w$ is the parent of $v$ and $v$ is a child of $w$.
A leaf is a vertex with no child.
Two vertices $v$ and $w$ are brothers if they have the same parent.
We consider ordered trees (also called plane trees),  meaning that
in the definition above 
a total order is assigned to each set of children of each vertex
(in our case, the total order is the order induced by 
the port numbering).
 Thus if we consider a vertex $v$ having $k$ children we can
speak of the first child, of the second child etc. If $v$ and $w$ are brothers,
let $u$ be their parent;
$w$ is said to be the next brother of $v$ if it is  the next 
successor of $v$
with respect to  the total
order assigned to the children of $u$.

%The height of a tree is the furthest distance between a node and the root
%plus one.

\subsection{Computing the Breadth-First-Search 
Spanning-Tree Rooted at 
Leader}
The first step of Algorithm $DEA$ computes a BFS spanning-tree.
Starting from Leader, the spanning-tree, denoted $BFS$-$ST$,
 is computed level by level by
the well-known procedure BFS (see \cite{Peleg} p. 50).

Initially each vertex is in the state $waiting$, in detail:
\begin{enumerate}
\item  Leader sends the signal $Start$ to all its neighbours;
\item any vertex in state $waiting$ receiving a $Start$ signal from one
or more neighbours at time $t$ does:
\begin{enumerate}
\item chooses as its parent the first such neighbour (in the order 
of enumeration  of its ports, for example);
\item at time $t+1$ sends $Accept$ to its parent and $Reject$ to the 
other neighbours that sent $Start$;
\item at time $t+1$ sends $Start$ to all its other neighbours;
\item at time $t+2$ notes as its child any neighbours sending $Accept$ and
sends  $Reject$ to all other neighbours that sent $Start$ at time $t+1$;
\end{enumerate}
\item as soon as a non-leader vertex has carried out step 2.(c)  and
 received $OK$ from all its 
children  (a leaf has no child), it sends $OK$ to its parent;
\item as soon as Leader  has sent $Start$ and 
received  $OK$ from all its children (that is all
its neighbours), it knows the breadth-first-search tree computation is complete.
\end{enumerate}

%\begin{remark}
%A wave stops as soon as it reachs the most distant vertices.
%\end{remark}

\begin{cl}
Let $G$ be a graph having $n$ vertices and a distinguished vertex $Leader$.
The procedure BFS  computes a BFS spanning-tree of $G$
in time $O(n)$.
Its bit complexity is also $O(n)$.
\end{cl}

We recall that
edges of $G$ that do not appear in the BFS spanning-tree connect
vertices either with the same level or with levels which differ by at most $1.$

\subsection{A Distributed  Enumeration Algorithm}
Once Leader knows that the BFS spanning-tree 
computation is complete, it starts a 
phase in which each vertex in turn (in an order to be
described later) starts a wave propagation which traverses the whole graph.
Thanks to the properties of the order, each vertex can calculate its distance
from any other vertex.
To define the order in which the waves are started, we  define
a traversal of the BFS spanning-tree and a numbering of vertices.

First we define a tree traversal, denoted $Trav$,
 used to visit and to number systematically
each vertex.

This tree traversal and the associated vertex numbering
may be defined iteratively as follows.
We add a loop on each leaf of the BFS spanning-tree  and  vertices are visited twice
 in a Depth-First-Search (DFS)
traversal (a leaf is visited on arriving and by following the loop). 

The traversal $Trav$ is defined by:
\begin{itemize}  
\item  if it is the first visit to a vertex then
go to the first child of the  vertex if it has a child; if it has no child
(i.e., it is a leaf) go from the leaf to itself;
\item if it is the second visit to a vertex go to the next
brother of the vertex,
if it has a next brother; if it has no next 
brother go to the parent of the vertex
if it has a parent else stop since it is the root of the BFS spanning-tree 
and the 
traversal is finished.
\end{itemize}
Thus, a visit to $v$ is immediately followed by a visit to a child or to 
a brother
or to the parent of $v$ or to $v$ itself (if $v$ is a leaf). 

\begin{cl}
Each vertex is visited twice.
\end{cl}

Let $v$ be a vertex, $\nu_v^{(1)}$ (resp. $\nu_v^{(2)}$)
 denotes the number of vertices visited
before the first visit to $v$ (resp. before the second visit to $v$).

By an induction on the level of vertices:
\begin{lemma}
Let $v$ be a vertex,
$\nu_v^{(1)}$ is even if and only if the level of $v$  is even;
$\nu_v^{(2)}$ is even if and only if the level of $v$  is odd.
\end{lemma}
%The position of a  visit to a vertex is the number of steps done before
%this visit.
From the  previous lemma:
\begin{corollary}
Let $v$ be a vertex. 
For each run of  $Trav$,  if $\nu_v^{(1)}$ is odd (resp. even) then
$\nu_v^{(2)}$ is even (resp. odd).
\end{corollary}
\noindent
The number of a vertex $v$ is obtained by 
computing the number of visited vertices  during the tree traversal
before the first or the second visit to $v.$
More precisely:
\begin{definition}
The vertex numbered $k$ is the $k^{th}$ visited 
vertex such that an even number of vertices have been visited before it;
it is denoted $v_k$.
\end{definition}

An example of a run of $Trav$ and the numbering of vertices
 is given in Fig. $1$.

\begin{figure}
\begin{center}
\epsfig{file=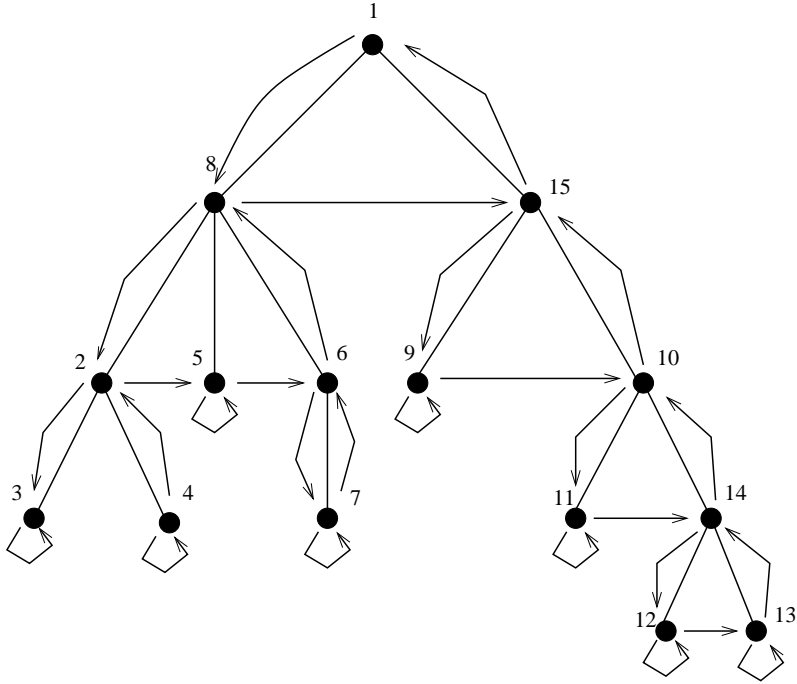}
\caption{An example of a run of $Trav$ with the associated
 numbering of the vertices.}
\end{center}
\end{figure}

By induction on the level of a vertex, we have:
\begin{lemma}
Let $v$ be the vertex numbered $k$. 
If $\nu_v^{(1)}$ is even then $k=\nu_v^{(1)}/2+1$
else  $k=\nu_v^{(2)}/2+1$.
\end{lemma}

Now, we can state the fundamental property of the numbering we use
later:
\begin{lemma}
Let $G$ be a connected graph. We consider the numbering
of vertices of $G$ after a run of $Trav$.
The distance between $v_i$ and $v_{i+1}$ is at most $3$
in the BFS spanning-tree and thus in $G.$
\end{lemma}
\begin{proof}
Two vertices having consecutive numbers are separated by two steps
in the traversal of the tree. 

Furthermore 
a visit to $v$ is immediately followed by a visit to a child or a brother
or the parent of $v$ or $v$ itself (if $v$ is a leaf). 
Therefore the distance between a vertex $v$ and the vertex $w$ reached after
two steps is at most $3$. The result follows.
\end{proof}
The numbering of the vertices can be achieved  by combining steps
of the traversal and sending messages $1$ (using unary representation
of the numbers) and $End$. More precisely:
\begin{enumerate}
\item Leader sends messages $1$ and $End$ on two successive
steps to its first child; the  number of Leader is $1$;
\item a vertex receiving $1$ sends $1$ to its successor with respect to
the $Trav$  traversal (i.e., to a child, a brother, the parent or
to itself (for a leaf))
on  the following step;
\item a vertex receiving $End$ sends $1$ and $End$ to its successor
with respect to the tree traversal on
the two following steps;
\item a vertex $v$ receives $1$ from two predecessors in traversal order;
let $p_1$ and $p_2$ be the numbers of $1$ received by $v$
from these two predecessors for this traversal;
one is even and the other  is odd. Assume $p_1$ is even; then
 the  number of $v$ is $p_1/2 +1$;
\item as soon as Leader has received $End$ from its last child
 it knows the numbering  is complete.
\end{enumerate}

As in previous sections:
\begin{cl}
Let $G$ be a connected graph having $n$ vertices.
Enumeration of the vertices of $G$ 
has a time complexity and a bit complexity
equal to $O(n).$
\end{cl}

\section{All Pairs Shortest Paths}
The computation of All Pairs Shortest paths needs, first, the computation
by each vertex of its distance to Leader.
\subsection{Calculating Distance from Leader to each Vertex}
Once Leader knows that the tree computation is complete, it starts
a process in which each other vertex learns its distance from Leader.
Leader sends a signal to each of its neighbours telling it that its distance
is $1$ and every vertex $v$ sends to all its children a message giving the 
child's 
distance as $1$ more than that of $v$. For simplicity, we describe a method
of achieving this using unary representation of the distances. Thus the
distance computation is obtained by the procedure $Dist$-$Cal$ defined by:
\begin{enumerate}
\item Leader sends messages $1$ and $End$ to each child on two successive
steps;
\item a vertex receiving $1$ sends $1$ to each of its children on
 the following step;
\item a vertex receiving $End$ sends $1$ and $End$ to its children on 
the two following steps;
\item the distance of a vertex from Leader is the number of $1$s received;
\item as soon as a non-leader vertex has sent $End$ to
and received OK from all its children
(including the case of a leaf which has no child), it sends OK to its parent;
\item as soon as Leader has received OK from all its children (that is
all its neighbours), it knows the distance computation is complete.
\end{enumerate}
\begin{cl}
Let $G$ be  a graph having $n$ vertices and a distinguished vertex $Leader$.
The procedure $Dist$-$Cal$ enables each vertex to know its distance to
Leader; the time complexity and the bit complexity required 
 are $O(n)$.
\end{cl}
\begin{remark}
As we consider a BFS spanning-tree, the level of a vertex is its 
distance to the root.
\end{remark}

\subsection{All Pairs Shortest Paths}
Once Leader knows that the enumeration of the vertices is completed and 
each node knows its distance to Leader, it 
starts a phase in which every vertex in the order of the enumeration
starts an anonymous  wave propagation for distance calculation.
The propagation of a wave follows these rules:
\begin{enumerate}
\item the source of the wave sends $wave$ once to all its neighbours;
\item a vertex $v$ receiving $wave$ at time $t$ from one or more neighbours
sends $wave$ at time $t+1$ to all other neighbours;
\item $v$ ignores any $wave$ signal received at time $t+1$ 
(from any neighbour at the same distance from the wave's source).
\end{enumerate}

Based on this mechanism, we give the  algorithm $APSP$ 
for  the all pairs shortest paths  calculation:
\begin{enumerate}
\item Leader starts a wave at time $t_1$;
\item each vertex computes $t_1$ when receiving the first wave;
\item each vertex $v_i$ starts a wave at time $t_i=t_1+5(i-1)$;
\item each vertex $v_i$ computes its distance to any vertex $v_j$
when it receives the $j^{th}$ wave;
\item each vertex $v$ computes the maximal distance to any vertex 
when no new wave arrives within eight steps after the last one.
\end{enumerate}

First:
\begin{cl}
Let $v$ be a vertex. If $v$ starts a wave following rules above then
each vertex $w$ receives $wave$ signals at times $d$ and (possibly) $d+1$
after the start of the propagation where $d$ is the distance  between $v$
and $w$.
\end{cl}
\begin{lemma} 
The waves start by two consecutive 
vertices (with respect to the enumeration order) won't collide with each other.
\end{lemma}
\begin{proof}
Let $t_1$ be the time at which $Leader$ starts its wave.
Every vertex knows its distance to $Leader$ and thus 
can compute $t_1$ as soon as it receives the signal $wave$ for the
first time.

Next, every vertex $v_i$ starts a new wave (denoted $w_i$)
 at time $t_i=t_1+5(i-1).$

Since the wave $w_{i+1 }$ starts $5$ rounds after the wave $w_i$ 
 and the distance between $v_i$ and $v_{i+1}$ is at most $3$, these two
waves arrive at any vertex $v$ separately and in the order $w_i$ followed
by  $w_{i+1}$
and at distance at least $2.$
\end{proof}

\begin{lemma}
Each vertex can compute its exact distance  to each other vertex.
\end{lemma}
\begin{proof}
Let $\tau_1$,...,$\tau_n$ denote the times of arrival of the $n$ waves
at $v.$ We have:
 $\tau_1=t_1+d(v,root)$.
Now,  $\tau_i=t_i+d(v,v_i)$
thus $d(v,v_i)=\tau_i-t_i=\tau_i-t_1-5(i-1),$ and $v$ can compute its distance
from $v_i.$

Hence when no new wave arrives at a vertex $v$ (i.e., no new  wave arrives 
within eight steps after the last one), $v$ knows its exact distance from
each other vertex.
\end{proof}

Finally:
\begin{theorem}
Let $G$ be a graph having $n$ vertices and a distinguished vertex. 
There exists a synchronous distributed
algorithm which computes APSP 
of $G$ in $O(n)$ rounds with a bit complexity equal to
$O(n)$. 
\end{theorem}
\section{Computing the Diameter}
This section indicates how to compute the diameter of $G$ by
centralising  the maximum distance and broadcasting the result.
\begin{theorem}
Let $G$ be a graph having $n$ vertices and a
distinguished vertex. 
There exists a synchronous distributed
algorithm which computes 
the diameter of $G$ in $O(n)$ rounds with a bit complexity equal to
$O(n)$. Furthermore each vertex knows the value of the diameter at the end
of the algorithm.
\end{theorem}
\begin{proof}

The eccentricities calculated in the previous section are now sent up 
via the 
BFS spanning-tree to its root. 

For a vertex  $v$, we write $m_v$ for 
the eccentricity of $v$.

Each vertex $v$ except the root will send to its parent the maximum distance,
denoted
$M_v$, from any vertex in the subtree rooted at $v$; this value will be sent in
$M_v+2$ consecutive rounds in the form of one $max$ signal followed by $M_v$ $1$
 signals
and one $endmax$ signal. A leaf can start this process as soon as it knows
its own maximum distance  $m_v$ since $M_v=m_v$.

A non-leaf vertex $v$ will wait until it has received the $max$ signal from each
of its children $w$. 
It then sends the $max$ signal to its parent and continues to send
$1$ signals until it has received the $endmax$ signal from each child. 
It now knows
the $M_w$  for each of its subtrees and  $m_v$ and so can compute 
 $M_v$ and send the required number of extra $1$
 signals to its parent. Note that
the number of $1$ signals already sent at this point is at most  $M_w$ 
for any child $w$ which sent the last $endmax$
 and so cannot be greater than $M_v$.

In this way the root knows the global maximum after diameter, and thus
at most a linear, number of rounds.

Finally, the root sends the global maximum (it suffices to use the same
unary format as previously) to each of its children who transmit it to each of their
children etc. Again this takes a linear number of rounds.
\end{proof}
Finally, the diameter is obtained by the following  steps:
\begin{enumerate}
\item Breadth-First-Search Tree Computation initiated by Leader;
\item Numbering of  vertices;
\item Calculating distance between Leader and each vertex;
\item Waves initiation and all pairs shortest paths 
calculation;
\item Centralisation of the maximum distance and broadcast of the diameter.
\end{enumerate}

\section{Other Applications of the Numbering of Vertices}
We illustrate  the power of the waves initiated by vertices with respect
to the numbering of vertices 
for computing girth and for the determination of cut-edges and cut-vertices.
\subsection{Computing the Girth}
\begin{theorem}
Let $G$ be a graph having $n$ vertices and a distinguished vertex.
The girth of $G$  can be computed by a distributed algorithm and
 known by each vertex with a time and a bit complexity equal to
$O(n)$.
\end{theorem}
\begin{proof}
As for the computation of the diameter, once Leader knows that the 
enumeration  of the vertices is complete it starts a phase in which every
vertex in the order of the enumeration starts an anonymous
 wave propagation.

Let $v$ be a vertex.
If $v$ receives the signal 
$wave$ from at least two neighbours at time $d$ then
it concludes that it belongs to a cycle of length $2d.$
If $v$ receives the signal $wave$ from a neighbour at time $d$ and the signal
$wave$ at time $d+1$ from another neighbour
then it concludes that it belongs to a cycle of length 
$2d+1.$

If $v$ belongs to a cycle then there is at least one vertex $u$
of this cycle such that $u$ starts a wave and this wave will reach
$v$ simultaneously by two different  edges incident to $v$ or will reach
$v$ by an edge incident to $v$ at time $d$ and by another edge incident to
$v$ at time $d+1$. Thus the length of this cycle will be calculated by $v$.

When  $v$ knows that no new wave will arrive (i.e., no new wave arrives within
eight steps after the last one) it computes the minimal length of  cycles
to which it belongs, denoted $c_v$.

If $v$ belongs to no cycle then, by convention the length is
$0.$

Now, each vertex sends to Leader $(n-c_v)$ by following the same 
procedure as in the previous section. It has found $n$ by counting the number
of waves.

Finally, as for the diameter, Leader centralises this value and it deduces
and transmits the girth.
\end{proof}
\subsection{Computing Cut-Edges}
Let $G$ be a connected graph; an edge is said to be a cut-edge if its 
deletion  disconnects $G.$

\begin{theorem}
Let $G$ be a connected graph having $n$ vertices and a distinguished vertex.
Cut-edges can be determined by a 
distributed algorithm  and known by endpoints of cut-edges in  $O(n)$
rounds with  a bit complexity equal to $O(n)$.
\end{theorem}
\begin{proof}
As for the computation of the diameter, once Leader knows that the 
enumeration  of the vertices is complete it starts a phase in which every
vertex in the order of the enumeration starts a wave propagation.

Let $v$ be a vertex. Let $e$ be an edge incident to $v$.
The theorem is a direct consequence of the following fact:

Edge $e$ is a cut-edge if and only if, whenever $v$ receives the signal
$wave$ through $e$ at time $t$, it does not receive it through
another edge at time $t$ or $t+1$.
\end{proof}
\subsection{Computing Cut-Vertices, Recognising Biconnected Graphs}
Let $G$ be a connected graph. Let $v$ be a vertex of $G$.
The vertex $v$ is a cut-vertex if removing $v$ and edges incident to $v$ 
disconnects $G$.
\begin{theorem}
Let $G$ be a connected graph with $n$ vertices and a distinguished vertex.
Cut-vertices of $G$ can be determined by a 
distributed algorithm   in  $O(n)$
rounds with a  bit complexity equal to $O(n)$.
\end{theorem}
\begin{proof}
As previously, once Leader knows that the 
enumeration  of the vertices is complete it starts a phase in which every
vertex in the order of the enumeration starts a wave propagation.

Let $u$ be a vertex having at least $2$ neighbours.
We define the relation $R_u$ as follows:
 two vertices $v_1$ and $v_2$ are  related modulo the 
relation $R_u$ if
the vertex $u$  receives the signal 
$wave$ from  $v_1$  at time $t$ and the signal $wave$ from  $v_2$
at time $t$ or at time $t+1$ for some time $t.$  

The theorem is a direct consequence of the following fact:
the vertex $u$ is not a cut-vertex if and only if every pair of
 neighbours of $u$
is  related modulo the transitive closure of $R_u$.
\end{proof}
This theorem has the following corollary:
\begin{corollary}
Let $G$ be a connected graph with $n$ vertices and a distinguished vertex.
To know whether G is biconnected 
can be determined by a distributed algorithm   in  $O(n)$
rounds with a  bit complexity equal to $O(n)$.
\end{corollary}
\begin{proof}
The result follows from the fact that a graph is biconnected
if and only if it has no cut-vertex.
\end{proof}
\section{Conclusion}
This work has been motivated by the distributed computation of all pairs 
shortest paths and the diameter with short messages. We  introduce a
distributed enumeration algorithm which uses messages of size $O(1)$.
From this enumeration we deduce algorithms for the computation
of all pairs shortest paths and
for the diameter which improve known results. 
Finally, this enumeration algorithm provides another proof of the fact that
the cube of a tree is a Hamiltonian graph and computes 
 a Hamiltonian cycle of the cube with a time complexity
and a bit complexity equal to $O(n)$.
\bibliographystyle{alpha}
\bibliography{biblio}

\end{document}